\documentclass[conference]{IEEEtran}

\usepackage{amsmath,amssymb,amsfonts,amsthm,mathtools}
\usepackage{algorithm}
\usepackage{algpseudocode}
\usepackage{array}
\usepackage{cite}
\usepackage{graphicx}
\usepackage[bookmarks=false]{hyperref}
\usepackage{svg}
\usepackage{xcolor}

\algnewcommand\algorithmicinput{\textbf{Input:}}
\algnewcommand\Input{
\item[\algorithmicinput]}
\algnewcommand\algorithmicoutput{\textbf{Output:}}
\algnewcommand\Output{
\item[\algorithmicoutput]}

\usepackage{booktabs}
\usepackage{subcaption}
\usepackage{caption}
\DeclareMathOperator{\diag}{diag}

\newtheorem{proposition}{Proposition}

\hypersetup{
  colorlinks=true,
  linkcolor=blue,
  filecolor=magenta,
  urlcolor=cyan,
  pdftitle={Adaptive c2-Perturbed AFDM Waveform Design for ISAC},
}

\title{\LARGE Adaptive $c_2$-Perturbed AFDM Waveform Design for Integrated Sensing and Communication}
\vspace{-4em}
\author{
  \IEEEauthorblockN{
    Shiqi Cui\IEEEauthorrefmark{1},
    Fan Zhang\IEEEauthorrefmark{2},
    Yuanshuo Gang\IEEEauthorrefmark{1},
    Zeping Sui\IEEEauthorrefmark{3},
    Tianqi Mao\IEEEauthorrefmark{1},
    Zhaocheng Wang\IEEEauthorrefmark{2}
  }
  \IEEEauthorblockA{
    \IEEEauthorrefmark{1}State Key Laboratory of Environment Characteristics and Effects for Near-space, Beijing Institute of Technology
  }
  \IEEEauthorblockA{
    \IEEEauthorrefmark{2}Department of Electronic Engineering, Tsinghua University
  }
  \IEEEauthorblockA{
    \IEEEauthorrefmark{3}School of Computer Science and Electronics Engineering, University of Essex, UK
  }
  \vspace{-3em}
}

\begin{document}

\maketitle

\begin{abstract}
  Affine frequency division multiplexing (AFDM) is a promising waveform for integrated sensing and communication (ISAC) systems owing to its superior performance in time--frequency doubly dispersive channels. However, AFDM still faces a pair of challenges: high PAPR and random data symbols produce imperfect autocorrelation sidelobes. To address these challenges, this paper proposes a real-time data-driven framework that optimizes the pre-chirp parameter $c_2$ to enhance the AFDM-ISAC performance. Specifically, a side-information-free optimization problem is formulated to reduce PAPR and the weighted integrated sidelobe levels of both aperiodic and periodic autocorrelation functions, with complexity comparable to that of the conventional AFDM receiver. Furthermore, an efficient non-monotone line-search spectral projected-gradient algorithm is developed by exploiting closed-form gradients. Simulation results demonstrate that the proposed method achieves a superior sensing vs. communications trade-off and is capable of striking a promoted bit error rate performance in the presence of severe power amplifier nonlinearity.
\end{abstract}

\section{Introduction}

As a foundational technology for future wireless networks, integrated sensing and communication (ISAC) necessitates waveforms that simultaneously support reliable communications and accurate sensing~\cite{ISAC_Survey, ISAC_Gang}. ISAC waveform researches have focused on adapting conventional waveforms from existing communication or sensing systems. Specifically, orthogonal frequency division multiplexing (OFDM) is widely regarded as one of the most representative ISAC waveforms~\cite{cross-domain}. However, in high-mobility scenarios, OFDM is highly sensitive to Doppler shifts, yielding significant inter-carrier interference (ICI)~\cite{AFDM}. To this end, affine frequency division multiplexing (AFDM) was proposed with inherent chirp signalling, resulting in resilience under doubly selective channels~\cite{AFDM, ISAC_Chirp_Sui, AFDM_Fan, SM_AFDM_Sui, AFDM_IM_TWC, HWI_AFDM_Sui}.

However, as an emerging waveform, AFDM still faces challenges in practical ISAC deployments. Specifically, the high peak-to-average power ratio (PAPR) significantly degrades its resilience to the nonlinearity of power amplifiers (PAs)~\cite {HWI_AFDM_Sui, DAFT, AFDM_GPS, choi_AFDM_PAPR, PAPR_AFDM_Meng}. Although existing OFDM PAPR reduction techniques are theoretically applicable to AFDM, these methods generally incur performance degradation\cite{choi_AFDM_PAPR}. Therefore, it is necessary to develop AFDM-tailored PAPR reduction algorithms. Moreover, it was illustrated in \cite{CP-OFDM} that OFDM is the optimal communication waveform for ranging, whereas AFDM remains suboptimal. This observation suggests that AFDM's sensing capability requires further improvement.

The authors of \cite{AFDM_IM_TWC} indicated that the AFDM pre-chirp parameter $c_2$ had a negligible effect on communication performance, which implies extra degrees of freedom in more advanced ISAC waveform design. Existing studies leveraged $c_2$ for index modulation and PAPR reduction~\cite{AFDM_IM_TWC, Radar_centric_AFDM, AFDM_GPS, choi_AFDM_PAPR}. In particular, a radar-centric AFDM waveform featuring chirp-domain index modulation was developed in \cite{Radar_centric_AFDM}, where chirp pairs were selected from a radar-compliant lookup table to maintain desirable ambiguity characteristics while embedding low-rate data. For PAPR reduction, grouped pre-chirp selection (GPS), which searched over a structured set of $c_2$ candidates, was proposed in \cite{AFDM_GPS}, whereas a per-slot efficient AFDM design with low PAPR was further developed in \cite{choi_AFDM_PAPR}. However, the above methods rely on discrete codebooks shared between the transmitter and the receiver, which either require side information or entail higher receiver complexity. Moreover, since the discrete codebooks comprise only a finite set of values, the PAPR solutions and the ambiguity functions of the above algorithms are suboptimal. These approaches predominantly focused on a single objective, thereby failing to achieve the joint optima of ISAC performance.

To address these limitations, we treat the parameter $c_2$ as a flexible variable and propose a continuous optimization framework based on subcarrier-wise perturbations of $c_2$ to jointly optimize the PAPR and the weighted integrated sidelobe level (WISL) of the autocorrelation function (ACF). Specifically, phase-safety constraints are introduced to bound the receiver-side phase distortion, thereby avoiding the need for side information transmission and additional receiver complexity. Moreover, the joint optimization problem is formulated by combining a smooth PAPR surrogate with the WISL metric. To efficiently solve the resulting constrained non-convex problem, we derive closed-form expressions for the gradients with respect to phase-domain variables and develop a non-monotone line-search spectral projected gradient (NMLS-SPG) algorithm. Simulations demonstrate that the proposed framework significantly mitigates PAPR and sidelobes compared to other baselines. Moreover, under PA nonlinearity, the proposed approach achieves superior communication performance over conventional AFDM.

\textit{Notation:} Vectors are denoted by bold lowercase letters, and matrices by bold uppercase letters. $(\cdot)^T$, $(\cdot)^H$, and $(\cdot)^*$ denote transpose, Hermitian transpose, and complex conjugate, respectively. $\mathbb{C}^{M\times N}$ and $\mathbb{R}^{N}$ denote the complex and real spaces. $\mathbf{I}$ and $\mathbf{0}$ denote the identity matrix and the all-zero vector or matrix with compatible dimensions, respectively. For a matrix $\mathbf{A}$ and a vector $\mathbf{a}$, $[\mathbf{A}]_{m,n}$ and $[\mathbf{a}]_m$ denote their $(m,n)$th and $m$th entries, respectively. Moreover, $\diag(\cdot)$ denotes a diagonal matrix, $\odot$ denotes the Hadamard product, $\mathbb{E}\{\cdot\}$ denotes expectation, $\Im\{\cdot\}$ denotes the imaginary part, and $\|\cdot\|_2$ and $\|\cdot\|_{\infty}$ denote the Euclidean and infinity norms.\vspace{-0.1cm}

\section{System Model}
\label{sec:system_model}

\subsection{Conventional AFDM}

Consider a conventional AFDM block whose symbol vector is $\mathbf{x} = [x_0,x_1,\ldots,x_{N-1}]^T \in \mathbb{C}^{N\times1}$. The corresponding transmitted signal in time-domain can be formulated as~\cite{AFDM}
\begin{equation}
  \mathbf{s} = \mathbf{A}^{H}\mathbf{x}
  =
  \boldsymbol{\Lambda}_{c_1}^{H}\mathbf{F}^{H}\boldsymbol{\Lambda}_{c_{2}}^{H}\mathbf{x},
  \label{eq:standard_afdm}
\end{equation}
where $\mathbf{F}$ is the discrete Fourier transform (DFT) matrix with entries $[\mathbf{F}]_{m,n} = \frac{1}{\sqrt{N}}e^{j2\pi mn/N}$, $c_1$ and $c_2$ denote the AFDM post-chirp and pre-chirp parameters, respectively. The chirp matrix is defined as
\begin{equation}
  \boldsymbol{\Lambda}_{c} = \diag\!\left(e^{-j2\pi c n^2}\right)_{n=0}^{N-1},
\end{equation}
where $c\in\{c_1,c_2\}$. Then, a chirp-periodic prefix (CPP) is inserted into the AFDM signal, which is then upconverted to radio frequency and passed through the PA. To account for transmitter PA nonlinearity, the equivalent baseband signal is modeled by a memoryless modified Rapp model~\cite{nokia2016realistic}, yielding
\begin{equation}
  \tilde{s}[n]
  =
  \frac{s[n]}{\left[1+\left(\frac{|s[n]|}{A_{\rm{sat}}}\right)^{2p}\right]^{\frac{1}{2p}}}
  \exp\!\left[j\kappa\frac{|s[n]|^{q}}{1+\left(\frac{|s[n]|}{A_{\phi}}\right)^{q}}\right],
  \label{eq:modified_rapp}
\end{equation}
for $n=0,\ldots,N-1$, where $A_{\rm{sat}}$ is the saturation level, $p$ controls the amplitude modulation to amplitude modulation (AM/AM) smoothness, and $\{\kappa,A_{\phi},q\}$ parameterize amplitude modulation to phase modulation (AM/PM) distortion. The severity of the nonlinearity is quantified by the input back-off (IBO) $10 \log_{10}(A_{\rm sat}^2/P_{\rm avg})$, where $P_{\rm avg}$ is the average input power. Equivalently, for a prescribed IBO, the saturation level is computed as $A_{\rm sat}=\sqrt{P_{\rm avg}10^{\mathrm{IBO}/10}}$. By Bussgang's theorem, the equivalent baseband signal in time-domain $\tilde{\mathbf{s}}=[\tilde{s}[0],\ldots,\tilde{s}[N-1]]^T$ can be expressed as~\cite{Bussgang}
\begin{equation}
  \tilde{\mathbf{s}} = \eta\mathbf{s}+\mathbf{q},
  \label{eq:bussgang_decomp}
\end{equation}
where $\eta=\mathbb{E}\!\left\{\mathbf{s}^{H}\tilde{\mathbf{s}}\right\} / \mathbb{E}\!\left\{\mathbf{s}^{H}\mathbf{s}\right\}$ is the Bussgang gain, and $\mathbf{q}$ represents nonlinear distortion, which has zero mean and satisfies $\mathbb{E}\left\{\mathbf{s}^{H}\mathbf{q}\right\}=0$.

After CPP removal, the doubly dispersive channel can be modeled as an $N \times N$ matrix $\mathbf{H}$. Then the received time-domain signal is
\begin{equation}
  \mathbf{r}
  =
  \mathbf{H}\tilde{\mathbf{s}} + \mathbf{n}
  =
  \eta\mathbf{H}\mathbf{s} + \mathbf{H}\mathbf{q}+\mathbf{n}
  \triangleq
  \eta \mathbf{H}_{\rm{eff}}\mathbf{x} + \tilde{\mathbf{n}},
  \label{eq:time_domain_receive}
\end{equation}
where $\mathbf{H}_{\rm{eff}} = \mathbf{HA}^H$ and $\tilde{\mathbf{n}}\triangleq\mathbf{H}\mathbf{q}+\mathbf{n}$ combines nonlinear distortion and additive white Gaussian noise (AWGN). Upon exploiting the minimum mean square error (MMSE) equalizer $\mathbf{W}=\mathbf{H}_{\rm{eff}}^H(\mathbf{H}_{\rm{eff}}\mathbf{H}_{\rm{eff}}^H+N_0\mathbf{I})^{-1}$, where $N_0$ denotes the AWGN noise power, the estimated symbol vector is given by
\begin{equation}
  \hat{\mathbf{x}} = \eta \mathbf{W}
  \mathbf{H}_{\rm{eff}}\mathbf{x} + \mathbf{W}\tilde{\mathbf{n}}.
  \label{eq:base_equalized_model}
\end{equation}

\subsection{Revised AFDM Scheme}
\label{subsec:AFDM_c2}
Compared with conventional AFDM, the proposed scheme applies subcarrier-wise perturbations to the transmitter $c_2$ parameter, while the receiver remains unchanged. Specifically, the scalar $c_2$ parameter is replaced by
\begin{equation}
  \mathbf{c}_{2,\rm{opt}} = \mathbf{c}_{2} + \Delta \mathbf{c}_2,
  \label{eq:c2_vector}
\end{equation}
where $\mathbf{c}_{2}$ is a constant vector whose entries are all equal to $c_{2}$. The corresponding transmitter-side diagonal matrix is
\begin{equation}
  \boldsymbol{\Lambda}_{c_{2,\rm{opt}}}^{H}
  =
  \diag\!\left(e^{j2\pi (c_{2}+\Delta c_{2,m})m^2}\right)_{m=0}^{N-1}
  =
  \boldsymbol{\Lambda}_{c_{2}}^{H}\boldsymbol{\Phi}^{H},
  \label{eq:lambda_c2_opt}
\end{equation}
where
\begin{equation}
  \boldsymbol{\Phi}^{H}
  =
  \diag\!\left(e^{j2\pi \Delta c_{2,m}m^2}\right)_{m=0}^{N-1}.
  \label{eq:Phi_def}
\end{equation}
Accordingly, the transmitted signal can be rewritten as
\begin{equation}
  \mathbf{s} = \mathbf{A}_{\rm{tx}}^{H}\mathbf{x}
  =
  \boldsymbol{\Lambda}_{c_1}^{H}\mathbf{F}^{H}\boldsymbol{\Lambda}_{c_{2,\rm{opt}}}^{H}\mathbf{x}
  =
  \boldsymbol{\Lambda}_{c_1}^{H}\mathbf{F}^{H}\boldsymbol{\Lambda}_{c_{2}}^{H}\boldsymbol{\Phi}^{H}\mathbf{x},
  \label{eq:tx_signal}
\end{equation}
which shows that the proposed $c_2$-perturbation scheme realizes the design as subcarrier-wise phase rotations on the transmitted waveform. Since the receiver remains unchanged, the same linear equalizer $\mathbf{W}$ as in \eqref{eq:base_equalized_model} is adopted. Substituting \eqref{eq:bussgang_decomp} and \eqref{eq:tx_signal} into the received-signal model yields
\begin{align}
  \hat{\mathbf{x}}
  &=
  \eta \mathbf{W}\mathbf{H}\boldsymbol{\Lambda}_{c_1}^{H}\mathbf{F}^{H}\boldsymbol{\Lambda}_{c_{2}}^{H}\boldsymbol{\Phi}^{H}\mathbf{x} + \mathbf{W}\tilde{\mathbf{n}} \nonumber \\
  &\triangleq \mathbf{T}\boldsymbol{\Phi}^{H}\mathbf{x} + \mathbf{W}\tilde{\mathbf{n}},
  \label{eq:equalized_model}
\end{align}
where $\mathbf{T} \triangleq \eta\mathbf{W}\mathbf{H}\boldsymbol{\Lambda}_{c_1}^{H}\mathbf{F}^{H}\boldsymbol{\Lambda}_{c_{2}}^{H}$. We decompose $\mathbf{T}$ as $\mathbf{T} = \diag(\mathbf{T}) + \Delta \mathbf{T} \triangleq \mathbf{T}^{\star} + \Delta \mathbf{T}$. Then, \eqref{eq:equalized_model} can be rewritten as
\begin{equation}
  \hat{\mathbf{x}} = \mathbf{T}^{\star}\boldsymbol{\Phi}^{H}\mathbf{x} + \Delta\mathbf{T}\boldsymbol{\Phi}^{H}\mathbf{x} + \mathbf{W}\tilde{\mathbf{n}}.
  \label{eq:T_decomposition}
\end{equation}
According to the central limit theorem, both $\Delta\mathbf{T}\boldsymbol{\Phi}^{H}\mathbf{x}$ and $\mathbf{W}\tilde{\mathbf{n}}$ can be approximated as AWGN. Therefore, the proposed $c_2$-perturbation scheme can be interpreted as introducing controlled phase distortion on the constellation.


\section{Joint PAPR--WISL Problem Formulation}
\label{sec:problem_formulation}
In this section, we first formulate the joint optimization problem by introducing a surrogate PAPR objective and a WISL objective under both aperiodic and periodic autocorrelation models. Then, a phase-safety constraint is imposed on the $c_2$ perturbation, leading to a normalized Pareto-weighted joint PAPR--WISL formulation in the phase domain. Finally, we analyze the impact of the $c_2$-perturbation on the communication performance and derive the modulation-aware safe-angle criteria for rectangular QAM constellations.

\subsection{Smoothed PAPR Objective}
Within each signal block $\mathbf{s}$, the PAPR is defined as
\begin{equation}
  J_{\rm{PAPR}}
  =
  \frac{P_{\rm{peak}}}{P_{\rm{avg}}}
  =
  \frac{N\|\mathbf{s}\|_{\infty}^{2}}{\mathbf{s}^{H}\mathbf{s}}.
  \label{eq:papr_exact}
\end{equation}
To obtain a differentiable approximation, we adopt the log-sum-exp surrogate. For any scalars $\{z_n\}_{n=0}^{N-1}$, we have
\begin{equation}
  \max_{0 \le n \le N-1} z_n
  \le
  \log\!\left(\sum_{n=0}^{N-1}e^{z_n}\right)
  \le
  \max_{0 \le n \le N-1} z_n + \log N.
  \label{eq:lse_property}
\end{equation}
Accordingly, the smoothed PAPR objective is expressed as
\begin{equation}
  J_{\rm{PAPR}}^{(\mu)}
  =
  \frac{1}{\mu}\log\!\left(\sum_{n=0}^{N-1}e^{\mu c_0 |s[n]|^2}\right),
  \label{eq:papr_surrogate_alpha}
\end{equation}
where $c_0 = \frac{N}{\|\mathbf{s}\|_2^2}$ and $\mu > 0$ is the smoothing parameter.

\subsection{Unified WISL Formulation}
The aperiodic ACF (AACF) and the periodic ACF (PACF) are considered, which correspond to the zero-padding (ZP) and CPP settings, respectively. Both the AACF and PACF can be expressed in the unified form as~\cite{CP-OFDM}
\begin{equation}
  r_k \triangleq
  \left\{
    \begin{aligned}
      &
      r_k^{(\mathrm{A})}
      =
      \mathbf{s}^{H}\mathbf{J}_k^{(\mathrm{A})}\mathbf{s}
      =
      \big(r_{-k}^{(\mathrm{A})}\big)^{*}, \\
      &
      r_k^{(\mathrm{P})}
      =
      \mathbf{s}^{H}\mathbf{J}_k^{(\mathrm{P})}\mathbf{s}
      =
      \big(r_{N-k}^{(\mathrm{P})}\big)^{*},
    \end{aligned}
    \right.
    \label{eq:acf_unified}
  \end{equation}
  for $k=0,1,\ldots,N-1$, where the shift matrices are defined as
  \begin{equation}
    \mathbf{J}_k^{(\mathrm{A})}
    =
    \begin{bmatrix}
      \mathbf{0} & \mathbf{I}_{N-k} \\
      \mathbf{0} & \mathbf{0}
    \end{bmatrix}, \quad
    \mathbf{J}_k^{(\mathrm{P})}
    =
    \begin{bmatrix}
      \mathbf{0} & \mathbf{I}_{N-k} \\
      \mathbf{I}_k & \mathbf{0}
    \end{bmatrix}.
    \label{eq:Jk_aperiodic}
  \end{equation}
  Consequently, the WISL objective is formulated as
  \begin{equation}
    J_{\rm{WISL}}^{(\nu)}
    =
    \sum_{k=1}^{N-1}\eta_k |\mathbf{s}^{H}\mathbf{J}_k^{(\nu)}\mathbf{s}|^2,
    \label{eq:J_WISL}
  \end{equation}
  where $\nu \in \{\mathrm{A},\mathrm{P}\}$ denotes the aperiodic or periodic case, respectively, and $\eta_k$ denotes the weighting factor.

  \subsection{Phase-Safety Constraints and Joint PAPR--WISL Problem}
  From \eqref{eq:Phi_def}, the perturbation on the $m$-th subcarrier induces the phase distortion
  \begin{equation}
    \theta_m \triangleq 2\pi m^2 \Delta c_{2,m}.
    \label{eq:theta_definition}
  \end{equation}
  To preserve reliable symbol detection, we impose the decision-safe phase constraint $|\theta_m| \le \theta_{\rm{safe}}$, which is equivalent to
  \begin{equation}
    |\Delta c_{2,m}| \le \frac{\theta_{\rm{safe}}}{2\pi m^2},
    \label{eq:funnel_constraint}
  \end{equation}
  for $m \in \{1,\ldots,N-1\}$, while $\theta_0=0$ holds trivially. Let us introduce a Pareto weight $\lambda \in [0,1]$. Since the PAPR and WISL have markedly different numerical scales, we normalize them by their zero-perturbation reference values so that $\lambda$ retains a clear physical interpretation. Specifically, the reference values can be formulated as
  \begin{equation}
    J_{\rm{PAPR}}^{\rm{ref}} \triangleq J_{\rm{PAPR}}^{(\mu)}(\Delta \boldsymbol{c}_2^{(0)}),
    \quad
    J_{\rm{WISL}}^{\rm{ref}} \triangleq J_{\rm{WISL}}^{(\nu)}(\Delta \boldsymbol{c}_2^{(0)}),
  \end{equation}
  where $\Delta \boldsymbol{c}_2^{(0)}=\mathbf{0}$ corresponds to conventional AFDM without perturbation. The normalized objective can be expressed as
  \begin{equation}
    \widetilde{J}(\Delta \boldsymbol{c}_2)
    =
    \lambda \frac{J_{\rm{PAPR}}^{(\mu)}(\Delta \boldsymbol{c}_2)}{J_{\rm{PAPR}}^{\rm{ref}}}
    +
    (1-\lambda)\frac{J_{\rm{WISL}}^{(\nu)}(\Delta \boldsymbol{c}_2)}{J_{\rm{WISL}}^{\rm{ref}}},
    \label{eq:Jmu_norm}
  \end{equation}
  where all the components of the objective are functions with respective to $\Delta \boldsymbol{c}_2$ through the transmitted signal $\mathbf{s}$. Accordingly, the joint optimization problem can be formulated as
  \begin{equation}
    \begin{aligned}
      \min_{\Delta \boldsymbol{c}_2}\quad
      & \widetilde{J}(\Delta \boldsymbol{c}_2) \\
      \text{s.t.}\quad
      & |\Delta c_{2,m}| \le \frac{\theta_{\rm{safe}}}{2\pi m^2},
      \quad m = 1,\ldots,N-1.
    \end{aligned}
    \label{eq:joint_problem_theta}
  \end{equation}
  As $\lambda \rightarrow 1$, the design becomes PAPR-dominant, whereas $\lambda \rightarrow 0$ prioritizes sidelobe suppression. This formulation also captures the implicit trade-off between bit error ratio (BER) and perturbation. Specifically, higher levels of $c_2$ perturbations induce larger phase rotations and may degrade BER performance.


  Since the induced distortion can be interpreted as a common phase rotation, the safety bound $\theta_{\rm{safe}}$ is determined by the geometry of the constellation. Consider a finite constellation $\mathcal{A}=\{A_1,\dots,A_M\}\subset\mathbb{C}$ under minimum-Euclidean-distance detection after the rotation $g=a e^{j\theta}$. Let the threshold $\theta_{\rm th}$ denotes the smallest positive angle which renders a rotated symbol $s_i e^{j\theta}$ reaches the boundary of the corresponding decision region. Given a margin $\gamma>0$, the threshold can be expressed as $\theta_{\rm th}=\theta_{\rm{safe}}+\gamma$.

  \begin{proposition}[Minimum unambiguous phase rotation of rectangular QAM]
    Consider a standard $N_I \times N_Q$ rectangular QAM constellation, where $M=N_I N_Q$ and $N_I \ge N_Q$. The alphabet of this constellation is given by $\mathcal{A}=\{a+j b:\; a\in\{\pm1,\pm3,\ldots,\pm(N_I-1)\},\; b\in\{\pm1,\pm3,\ldots,\pm(N_Q-1)\}\}$.
    The threshold $\theta_{\rm th}$ can be expressed as
    \begin{equation}
      \theta_{\rm th}
      =
      \arctan\!\left(\frac{N_Q-1}{N_I-1}\right) -
      \arcsin\!\left(
        \frac{N_Q-2}{d_{\rm{max}}}
      \right),
      \label{eq:theta_th_rect}
    \end{equation}
    where $d_{\rm{max}} = \sqrt{(N_I-1)^2+(N_Q-1)^2}$. Specifically, for the square QAM case with $N_I=N_Q=\sqrt{M}$, we have
    \begin{equation}
      \theta_{\rm th}^{\rm sq}(M)
      =
      \arccos\!\left(
        \frac{\sqrt{M}-2}{\sqrt{2}\,(\sqrt{M}-1)}
      \right)
      -\frac{\pi}{4}.
    \end{equation}
  \end{proposition}
  \begin{proof}
    Given the range of phase distortion of $0<\theta<\pi/2$, the most related symbol is given by $A_c=(N_I-1)-j(N_Q-1)$. This is because the horizontal decision margin of a generic point $(a-jb)$ is given by
    \begin{equation}
      \Delta(a,b;\theta)=b\cos\theta-a\sin\theta-(b-1),
    \end{equation}
    which decreases monotonically with respect to both $a$ and $b$. Accordingly, the solution to $\Delta(N_I-1,N_Q-1;\theta)=0$ is exactly the threshold of \eqref{eq:theta_th_rect} under the case of rectangular QAM. Consequently, the result for square QAM is obtained by setting $N_I=N_Q=\sqrt{M}$.
  \end{proof}

  \section{Proposed NMLS-SPG Algorithm}
  \label{sec:proposed_method}
  In this section, we propose the NMLS-SPG algorithm to solve the normalized phase-safety-constrained joint optimization problem of \eqref{eq:joint_problem_theta}. Specifically, we first derive the analytical gradient of the normalized objective function \eqref{eq:Jmu_norm}, and then derive projected spectral updates relying on a non-monotone line-search strategy. Furthermore, the computational complexity of the NMLS-SPG algorithm under both aperiodic and periodic WISL settings is quantified using fast Fourier transform (FFT)-based structured operators.
  \subsection{Closed-Form Gradient of the Normalized Objective}

  For the sake of notational simplicity, the superscripts $(\mu)$ and $(\nu)$ associated with the two components in \eqref{eq:Jmu_norm} are omitted throughout this subsection, i.e., we set $J_{\rm{PAPR}} \triangleq J_{\rm{PAPR}}^{(\mu)}$ and $J_{\rm{WISL}} \triangleq J_{\rm{WISL}}^{(\nu)}$. Furthermore, by adopting the induced phase vector $\boldsymbol{\theta} \triangleq [\theta_0,\theta_1,\ldots,\theta_{N-1}]^T$, whose elements $\theta_m$ are defined in \eqref{eq:theta_definition}, the objective function can be equivalently transformed into the phase domain. Consequently, the smoothed joint objective function is expressed as
  \begin{equation}
    J(\boldsymbol{\theta}) = \lambda J_{\rm{PAPR}}(\boldsymbol{\theta}) + (1-\lambda) J_{\rm{WISL}}(\boldsymbol{\theta}),
  \end{equation}
  Define the phase vector
  $\boldsymbol{\phi}(\boldsymbol{\theta}) \triangleq [e^{j\theta_0},e^{j\theta_1},\ldots,e^{j\theta_{N-1}}]^T$. Then, from \eqref{eq:tx_signal}, the transmitted signal can be written as $\mathbf{s}(\boldsymbol{\theta}) = \mathbf{A}^{H}\diag(\mathbf{x})\boldsymbol{\phi}(\boldsymbol{\theta}) \triangleq \mathbf{G}\boldsymbol{\phi}(\boldsymbol{\theta})$. Let $\mathbf{g}_n$ denote the $n$-th column of $\mathbf{G}^{H}$, we have $s_n(\boldsymbol{\theta}) = \mathbf{g}_n^{H}\boldsymbol{\phi}(\boldsymbol{\theta})$ and
  $q_n(\boldsymbol{\theta}) \triangleq |s_n(\boldsymbol{\theta})|^2 = \boldsymbol{\phi}(\boldsymbol{\theta})^{H}\mathbf{g}_n\mathbf{g}_n^{H}\boldsymbol{\phi}(\boldsymbol{\theta})$. Hence, the Wirtinger derivative of $q_n(\boldsymbol{\theta})$ with respect to $\boldsymbol{\phi}^{*}$ is given by
  $\nabla_{\boldsymbol{\phi}^{*}} q_n(\boldsymbol{\theta}) = \mathbf{g}_n\mathbf{g}_n^{H}\boldsymbol{\phi}(\boldsymbol{\theta}) = s_n(\boldsymbol{\theta}) \, \mathbf{g}_n$. Then, the gradient of the PAPR term in \eqref{eq:papr_surrogate_alpha} with respect to $\boldsymbol{\phi}^{*}$ can be formulated as
  \begin{align}
    \nabla_{\boldsymbol{\phi}^*} J_{\rm PAPR}(\boldsymbol{\theta})
    &=
    c_0 \sum_{n=0}^{N-1}
    \underbrace{\frac{\exp(\mu c_0 q_n(\boldsymbol{\theta}))}{\sum_\ell \exp(\mu c_0 q_\ell(\boldsymbol{\theta}))}}_{\triangleq \, \alpha_n(\boldsymbol{\theta})}
    s_n(\boldsymbol{\theta}) \, \mathbf{g}_n, \nonumber \\
    &=
    c_0 \mathbf{G}^H \big(\boldsymbol{\alpha}(\boldsymbol{\theta}) \odot \mathbf{s}(\boldsymbol{\theta})\big),
    \label{eq:grad_papr}
  \end{align}
  where $\boldsymbol{\alpha}(\boldsymbol{\theta}) = [\alpha_0(\boldsymbol{\theta}),\ldots,\alpha_{N-1}(\boldsymbol{\theta})]^T$. Similarly, regarding the WISL in \eqref{eq:J_WISL}, by introducing the auxiliary variable $\rho_k(\boldsymbol{\theta}) \triangleq \mathbf{s}(\boldsymbol{\theta})^{H}\mathbf{J}_k\mathbf{s}(\boldsymbol{\theta})$, the gradient of the WISL term can be expressed as
  \begin{align}
    \nabla_{\boldsymbol{\phi}^{*}}& J_{\rm{WISL}}(\boldsymbol{\theta})
    =
    \sum_{k=1}^{N-1} \eta_k \, \nabla_{\boldsymbol{\phi}^{*}} |\rho_k(\boldsymbol{\theta})|^2, \nonumber \\
    &=
    \mathbf{G}^{H}
    \sum_{k=1}^{N-1}
    \eta_k
    \big(\rho_k(\boldsymbol{\theta})^{*}\mathbf{J}_k + \rho_k(\boldsymbol{\theta}) \mathbf{J}_k^{H}\big)\mathbf{s}(\boldsymbol{\theta}).
    \label{eq:grad_wisl}
  \end{align}

  Based on the reference values $J_{\rm{PAPR}}^{\rm{ref}}$ and $J_{\rm{WISL}}^{\rm{ref}}$, the corresponding normalized combined vector is given by
  \begin{equation}
    \begin{aligned}
      \widetilde{\mathbf{v}}(\boldsymbol{\theta}) &
      \triangleq
      \mathbf{G}^{H}\!\Bigg[
        \frac{\lambda c_0}{J_{\rm{PAPR}}^{\rm{ref}}}
        \big(\boldsymbol{\alpha}(\boldsymbol{\theta})\odot\mathbf{s}(\boldsymbol{\theta})\big) \\
        & + \frac{1-\lambda}{J_{\rm{WISL}}^{\rm{ref}}}
        \sum_{k=1}^{N-1}\eta_k \, \big(\rho_k(\boldsymbol{\theta})^{*}\mathbf{J}_k
        + \rho_k(\boldsymbol{\theta})\mathbf{J}_k^{H}\big)\mathbf{s}(\boldsymbol{\theta})
      \Bigg].
    \end{aligned}
    \label{eq:v_combined}
  \end{equation}
  Applying the chain rule in the phase domain, the real-valued gradient of the normalized objective with respect to $\boldsymbol{\theta}$ is explicitly formulated as
  \begin{align}
    \nabla_{\boldsymbol{\theta}}\widetilde{J}(\boldsymbol{\theta})
    &=
    j\,\boldsymbol{\phi}(\boldsymbol{\theta})\odot\widetilde{\mathbf{v}}(\boldsymbol{\theta})^{*}
    -
    j\,\boldsymbol{\phi}(\boldsymbol{\theta})^{*}\odot\widetilde{\mathbf{v}}(\boldsymbol{\theta}) \nonumber \\[3pt]
    &=
    2\,\Im\!\left\{\boldsymbol{\phi}(\boldsymbol{\theta})^{*}\odot\widetilde{\mathbf{v}}(\boldsymbol{\theta})\right\}.
    \label{eq:grad_theta}
  \end{align}

  \addtolength{\topmargin}{0.051in}
  \begin{algorithm}[t]
    \caption{NMLS-SPG for $c_2$-perturbation optimization}
    \label{alg:nmls_spg}
    \footnotesize
    \begin{algorithmic}[1]
      \Input Initial phase $\boldsymbol{\theta}^{(0)}=\mathbf{0}$, WISL mode $\nu$, weight $\lambda$, smoothing parameter $\mu$, weights $\{\eta_k\}$, memory length $M$, line-search factor $\delta$, spectral bounds $(\beta_{\min},\beta_{\max})$, Armijo parameter $c$, phase bound $\theta_{\rm{safe}}$, maximum iterations $T_{\max}$, tolerance $\varepsilon$
      \Output Optimized phase vector $\boldsymbol{\theta}^{\star}$
      \State Compute the reference values $J_{\rm{PAPR}}^{\rm{ref}}=J_{\rm{PAPR}}(\boldsymbol{\theta}^{(0)})$ and $J_{\rm{WISL}}^{\rm{ref}}=J_{\rm{WISL}}(\boldsymbol{\theta}^{(0)})$
      \State Compute $\mathbf{g}^{(0)} = \nabla_{\boldsymbol{\theta}}\widetilde J(\boldsymbol{\theta}^{(0)})$
      \State Set $\beta_0 = 1/\max\{\|\mathbf{g}^{(0)}\|_2,\epsilon_0\}$ and initialize $f_{\rm{hist}}(1)=\widetilde J(\boldsymbol{\theta}^{(0)})$
      \For{$t=0,1,\ldots,T_{\max}-1$}
      \State $\bar{\boldsymbol{\theta}}=\Pi_{\mathbf{\Theta}}\!\left(\boldsymbol{\theta}^{(t)}-\beta_t\mathbf g^{(t)}\right)$
      \State $\mathbf{d}^{(t)} = \bar{\boldsymbol{\theta}} - \boldsymbol{\theta}^{(t)}$
      \If{$(\mathbf{g}^{(t)})^{T}\mathbf{d}^{(t)} \ge 0$}
      \State $\mathbf d^{(t)}=-\mathbf g^{(t)}$; \quad $d^{(t)}_0=0$
      \EndIf
      \State Set $f_{\max}^{(t)} = \max(f_{\rm{hist}})$, $\tau=1$
      \While{\eqref{eq:gll_condition} is not satisfied}
      \State Update $\tau \gets \delta \tau$
      \EndWhile
      \State $\boldsymbol{\theta}^{(t+1)} = \boldsymbol{\theta}^{(t)} + \tau \mathbf{d}^{(t)}$
      \State $\mathbf{g}^{(t+1)} = \nabla_{\boldsymbol{\theta}}\widetilde J(\boldsymbol{\theta}^{(t+1)})$
      \State $\mathbf{u}_t = \boldsymbol{\theta}^{(t+1)} - \boldsymbol{\theta}^{(t)}$, $\mathbf{y}_t = \mathbf{g}^{(t+1)} - \mathbf{g}^{(t)}$
      \If{$\mathbf{u}_t^{T}\mathbf{y}_t > 0$}
      \State $\beta_{t+1}^{\rm{BB}} = \mathbf{u}_t^{T}\mathbf{u}_t / (\mathbf{u}_t^{T}\mathbf{y}_t)$; \\ $\qquad \quad \beta_{t+1}\leftarrow \min\{\beta_{\max},\max\{\beta_{\min},\beta_{t+1}^{\rm{BB}}\}\}$
      \Else
      \State $\beta_{t+1}\leftarrow \min\{\beta_{\max},\max\{\beta_{\min},2\beta_t\}\}$
      \EndIf
      \State $f_{\rm hist}( (t \bmod M) + 1 ) = \widetilde J(\boldsymbol{\theta}^{(t+1)})$
      \If{$t+1\ge M$ and $\max(f_{\rm{hist}})-\min(f_{\rm{hist}})\le \varepsilon$}
      \State \textbf{break}
      \EndIf
      \EndFor
      \State \Return $\boldsymbol{\theta}^{\star} = \boldsymbol{\theta}^{(t+1)}$
    \end{algorithmic}
  \end{algorithm}

  \begin{figure*}[t]
    \centering
    \makebox[\textwidth][c]{%
      \begin{subfigure}[b]{0.34\textwidth}
        \centering
        \includegraphics[width=\textwidth]{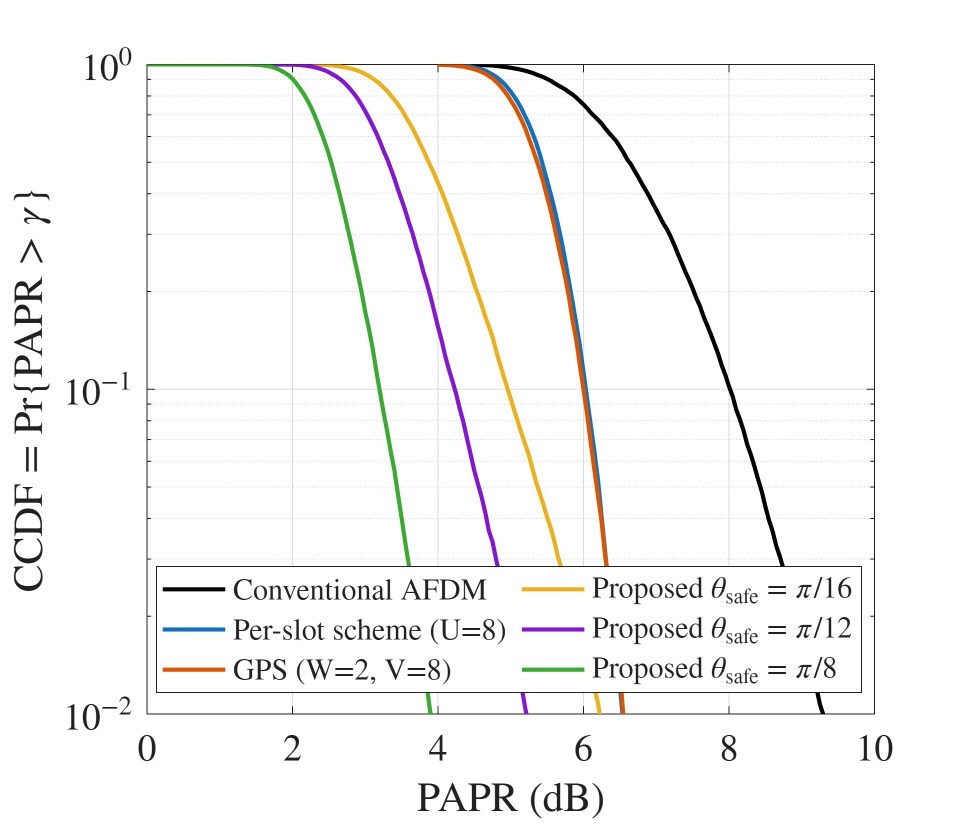}
        \caption{CCDF of PAPR.}
        \label{fig:PAPR}
      \end{subfigure}
      \hspace{-4mm}
      \begin{subfigure}[b]{0.34\textwidth}
        \centering
        \includegraphics[width=\textwidth]{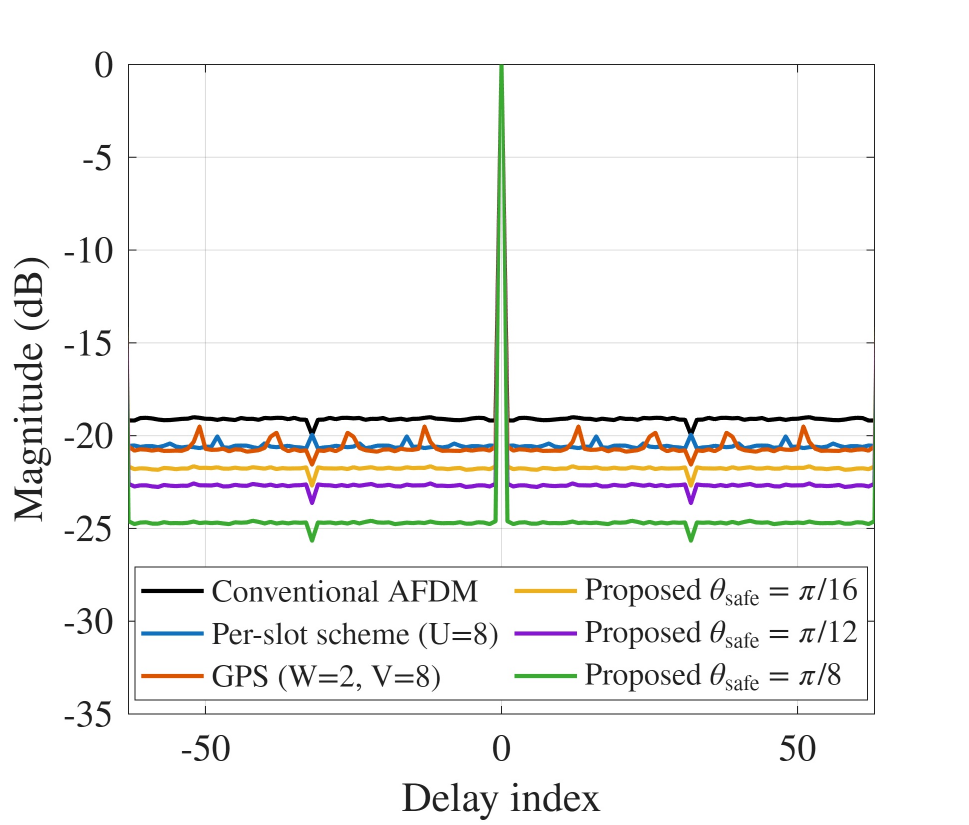}
        \caption{Squared magnitude of normalized PACF.}
        \label{fig:PACF}
      \end{subfigure}
      \hspace{-4mm}
      \begin{subfigure}[b]{0.32\textwidth}
        \centering
        \includegraphics[width=\textwidth]{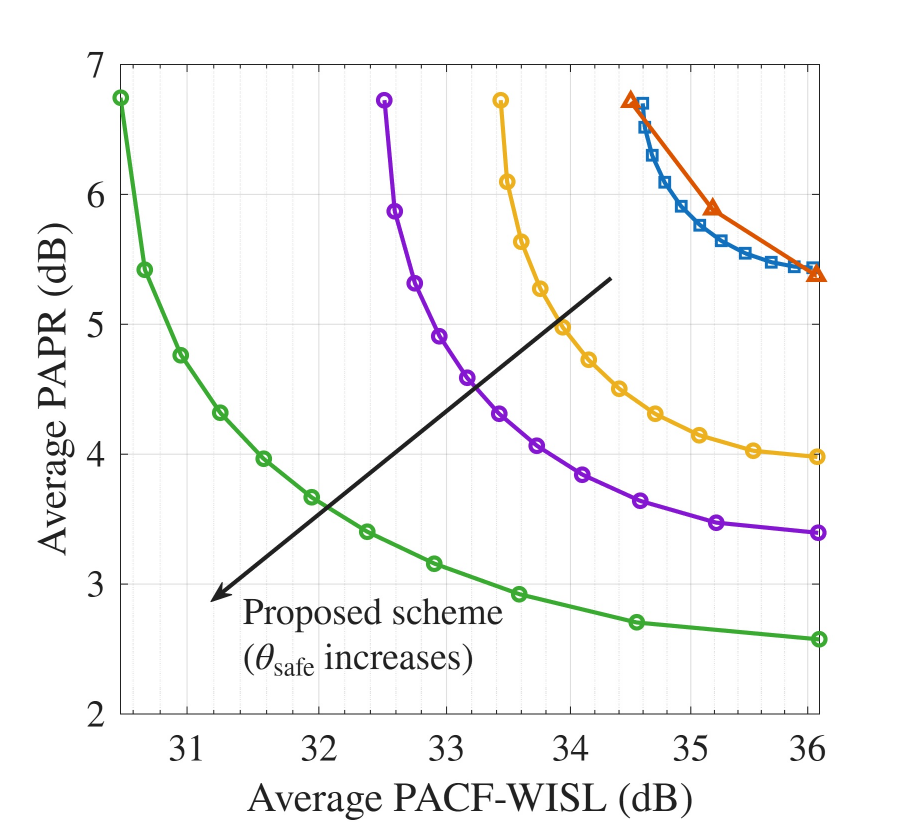}
        \caption{Pareto front.}
        \label{fig:results_Pareto}
      \end{subfigure}
    }
    \caption{Objective-performance results.}
    \label{fig:results_objective}
    \vspace{-5mm}
  \end{figure*}

  \subsection{NMLS-SPG Iterative Updates}
  The problem formulated in \eqref{eq:joint_problem_theta} is a smooth, non-convex optimization problem with simple box constraints, making first-order projected methods appealing because they exploit the closed-form gradient while preserving low-complexity via the FFT. At iteration $t$, define $\mathbf{g}^{(t)} = \nabla_{\boldsymbol{\theta}}\widetilde J(\boldsymbol{\theta}^{(t)})$. For $t\ge 1$, the Barzilai–Borwein (BB) spectral step-size is
  \begin{equation}
    \beta_t^{\rm{BB}}
    =
    \frac{(\Delta \boldsymbol{\theta})^{T}\Delta \boldsymbol{\theta}}
    {(\Delta \boldsymbol{\theta})^{T}\Delta \mathbf{g}},
    \label{eq:bb_step}
  \end{equation}
  where $\Delta \boldsymbol{\theta} = \boldsymbol{\theta}^{(t)}-\boldsymbol{\theta}^{(t-1)}$ and $\Delta \mathbf{g} = \mathbf{g}^{(t)} - \mathbf{g}^{(t-1)}$. To improve numerical robustness, the practical step-size is safeguarded via
  \begin{equation}
    \beta_t
    =
    \min\!\left\{\beta_{\max},\max\!\left\{\beta_{\min},\beta_t^{\rm{BB}}\right\}\right\},
    \label{eq:bb_safe}
  \end{equation}
  where $\beta_{\min}$ and $ \beta_{\max}$ denote the lower and upper safeguards on the spectral step-size. We initialize $\beta_0$ as $\beta_0 = 1 / \max\{\|\mathbf{g}^{(0)}\|_2,\epsilon_0\}$, where $\epsilon_0$ denotes a sufficiently small constant to prevent numerical overflow. Then, the projected search direction is
  \begin{equation}
    \mathbf{d}^{(t)}
    =
    \Pi_{\mathbf{\Theta}}\!\left(\boldsymbol{\theta}^{(t)} - \beta_t \mathbf{g}^{(t)}\right)
    - \boldsymbol{\theta}^{(t)},
    \label{eq:projected_direction}
  \end{equation}
  where $\mathbf{\Theta}$ denotes the feasible phase set
  \begin{equation}
    \mathbf{\Theta}
    =
    \left\{
      \boldsymbol{\theta}\in\mathbb{R}^{N}:
      |\theta_m| \le \theta_{\rm{safe}},\,
      m=1,\ldots,N-1
    \right\},
    \label{eq:feasible_set}
  \end{equation}
  and $\Pi_{\mathbf{\Theta}}(\cdot)$ denotes Euclidean projection onto $\mathbf{\Theta}$, with the following element-wise closed form
  \begin{equation}
    [\Pi_{\mathbf{\Theta}}(\mathbf{z})]_m = \operatorname{clip}(z_m;-\theta_{\rm{safe}},\theta_{\rm{safe}}), \ m=1,\ldots,N-1.
    \label{eq:proj_closed_form}
  \end{equation}
  where $\operatorname{clip}(z;a,b) \triangleq \min\{b,\max\{a,z\}\}$.

  In highly non-convex problems such as the proposed joint PAPR-WISL phase optimization, a strictly monotonic line search severely limits the effectiveness of the BB step. Consequently, to preserve the aggressive spectral step-size within the non-convex setting, we adopt the Grippo-Lampariello-Lucidi (GLL) non-monotone reference value
  \begin{equation}
    f_{\max}^{(t)}
    =
    \max_{0 \le j \le \min(t,M-1)}
    \widetilde J\big(\boldsymbol{\theta}^{(t-j)}\big),
    \label{eq:fmax}
  \end{equation}
  with memory length $M \ge 1$. This non-monotone window is implemented by a cyclic history buffer fhist, whose entries store the most recent $\min(t+1, M)$ objective values of $\widetilde J$, so that $f_{\max}^{(t)} = \max\big(f_{\rm{hist}}\big).$ The accepted step length $\tau \in (0,1]$ is chosen to satisfy
  \begin{equation}
    \widetilde J\big(\boldsymbol{\theta}^{(t)} + \tau \mathbf{d}^{(t)}\big)
    \le
    f_{\max}^{(t)} + c\tau (\mathbf{g}^{(t)})^{T}\mathbf{d}^{(t)},
    \label{eq:gll_condition}
  \end{equation}
  where $c \in (0,1)$ is the Armijo parameter. The length-$M$ non-monotone history buffer is then updated cyclically, and the iteration terminates once the objective span within this buffer falls below $\varepsilon$. Our NMLS-SPG is summarized in \textbf{Algorithm}~\ref{alg:nmls_spg}.

  \subsection{Computational Complexity}

  In the proposed algorithm, the computation of $\mathbf{s}=\mathbf{G}\boldsymbol{\phi}$ requires one length-$N$ IFFT, while the computation of $\mathbf{G}^{H}$ requires one length-$N$ FFT. Both PACF and AACF can be efficiently evaluated via the Wiener--Khinchin theorem. For PACF, the circular autocorrelation is given by $\tilde{\mathbf{r}}=\mathbf{F}^{H}\big(|\mathbf{F}\mathbf{s}|^2\big)$. For AACF, zero padding is first applied to obtain $\mathbf{s}_{\rm{zp}}=[\mathbf{s}^{T},\mathbf{0}_{N}^{T}]^{T}\in\mathbb{C}^{2N}$, and the linear autocorrelation is computed as $\mathbf{r}_{\rm{lin}}=\mathbf{F}_{2N}^{H}\!\big(|\mathbf{F}_{2N}\mathbf{s}_{\rm{zp}}|^2\big)$. All remaining operations exhibit linear complexity with the order of $\mathcal{O}(N)$. Let $C_{\rm{FFT}}(L) = L\log_2 L$ denotes the complexity of a length-$L$ FFT or IFFT operation. By defining the autocorrelation transform length as $M=N$ for PACF and $M=2N$ for AACF, the costs of the objective and gradient evaluations can be respectively formulated as\vspace{-0.1cm}
  \begin{align}
    C_{J}(N) &= C_{\rm{FFT}}(N)+2C_{\rm{FFT}}(M)+\mathcal{O}(N), \\
    C_{\nabla}(N) &= 2C_{\rm{FFT}}(N)+4C_{\rm{FFT}}(M)+\mathcal{O}(N).
    \label{eq:complexity_compact}
  \end{align}
  The total complexity can be compactly expressed as\vspace{-0.1cm}
  \begin{align}
    C_{\rm{tot}} &= (2T+S_T)\big[C_{\rm{FFT}}(N)+2C_{\rm{FFT}}(M)\big] + \mathcal{O}(TN) \nonumber \\
    &= \mathcal{O}\!\left(T(1+\bar L)N\log N\right),
    \label{eq:complexity_total_final}
  \end{align}
  where $T$ denotes the number of outer iterations, $L_t\ge 1$ is the number of objective evaluations in the GLL line search at iteration $t$, $S_T\triangleq\sum_{t=1}^{T}L_t$, and $\bar L\triangleq S_T/T$.

  \section{Simulation Results}
  \label{sec:simulation_results}
  In this section, we characterize the performance of the proposed $c_2$-perturbation scheme in terms of waveform peak reduction, sensing sidelobe suppression, and communication reliability. Unless otherwise specified, the evaluations consider an AFDM system with $N=64$ subcarriers, $c_1=\frac{5}{2N}$, $c_2=\sqrt{2}$, and QPSK modulation. For the transmitter PA, we adopt the modified Rapp model presented in \eqref{eq:modified_rapp}. This model is adapted from~\cite{nokia2016realistic} to represent a PA with unit gain, where the parameters are defined as $\{p, \kappa, A_\phi, q\} = \{2, -0.315, 1.137, 4\}$ and $A_{\rm sat}$ is determined based on the IBO values $\mathrm{IBO}\in\{0~\mathrm{dB}, 10~\mathrm{dB}\}$. The considered benchmarks include conventional AFDM, the per-slot scheme, and GPS. To ensure a fair comparison, the benchmarks adopt the configurations provided in \cite{choi_AFDM_PAPR, AFDM_GPS}. For \textbf{Algorithm~\ref{alg:nmls_spg}}, we set $\nu=\mathrm{P}$ to use the PACF-based WISL metric, while $\lambda$ is uniformly distributed over $[0:0.1:1]$ with $\mu=20$, uniform weights $\{\eta_k\}$, $M=10$, and $\delta=0.5$. The spectral step-size is constrained by $\beta_{\min}=10^{-5}$ and $\beta_{\max}=10^{5}$, the Armijo parameter is set to $c=10^{-4}$, the phase-safety threshold is selected from $\{\pi/16,\pi/12,\pi/8\}$, and the stopping parameters are given by $T_{\max}=20$ and $\varepsilon=10^{-5}$.


  As illustrated in Fig.~\ref{fig:PAPR}, the complementary cumulative distribution function (CCDF) of the PAPR of our proposed scheme outperforms all the baselines relative to conventional AFDM. This is because the optimization space of the baseline schemes is constrained by discrete codebooks, whereas the proposed method employs continuous optimization to achieve a superior solution. Moreover, the proposed method attains the highest gain when $\theta_{\rm{safe}}=\pi/8$. Specifically, relative to conventional AFDM, the $\theta_{\rm{safe}}=\pi/8$ curve achieves an approximate $4$--$5$~dB PAPR reduction at a CCDF near $10^{-2}$, while outperforming both reference optimization schemes.
  
  Fig.~\ref{fig:PACF} compares the normalized PACF sidelobes. It can be observed that the PACF magnitude of conventional AFDM stays near $-19$~dB over most nonzero delay indices, whereas the proposed design's sidelobe floor decreases progressively as $\theta_{\rm{safe}}$ increases. When $\theta_{\rm{safe}}=\pi/8$, the average sidelobe level decreases to approximately $-25$~dB, indicating that the proposed solution can mitigate WISL and promote the correlation behavior for sensing.

  As shown in Fig.~\ref{fig:results_Pareto}, the Pareto front of the proposed method outperforms the per-slot and GPS methods, which are represented by the red and dark-blue curves, respectively, indicating a superior trade-off between PACF-WISL and PAPR in the considered performance region. Furthermore, by increasing the value of $\theta_{\rm{safe}}$, the achievable region can be expanded towards lower values of both PACF-WISL and PAPR, confirming that the phase-safety bound directly governs the attainable performance envelope.

  \begin{figure}[t]
    \centering
    \includegraphics[width=0.8\columnwidth]{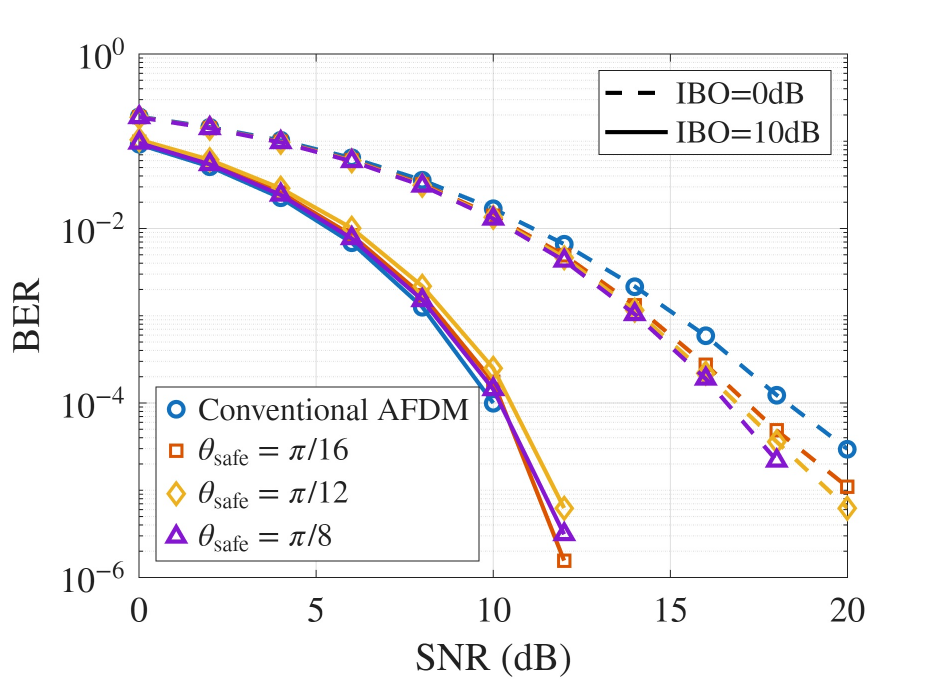}
    \caption{BER performance for BPSK with different $\theta_{\rm{safe}}$.}
    \label{fig:results_BER}
    \vspace{-4mm}
  \end{figure}

  \begin{figure}[t]
    \centering
    \includegraphics[width=0.8\columnwidth]{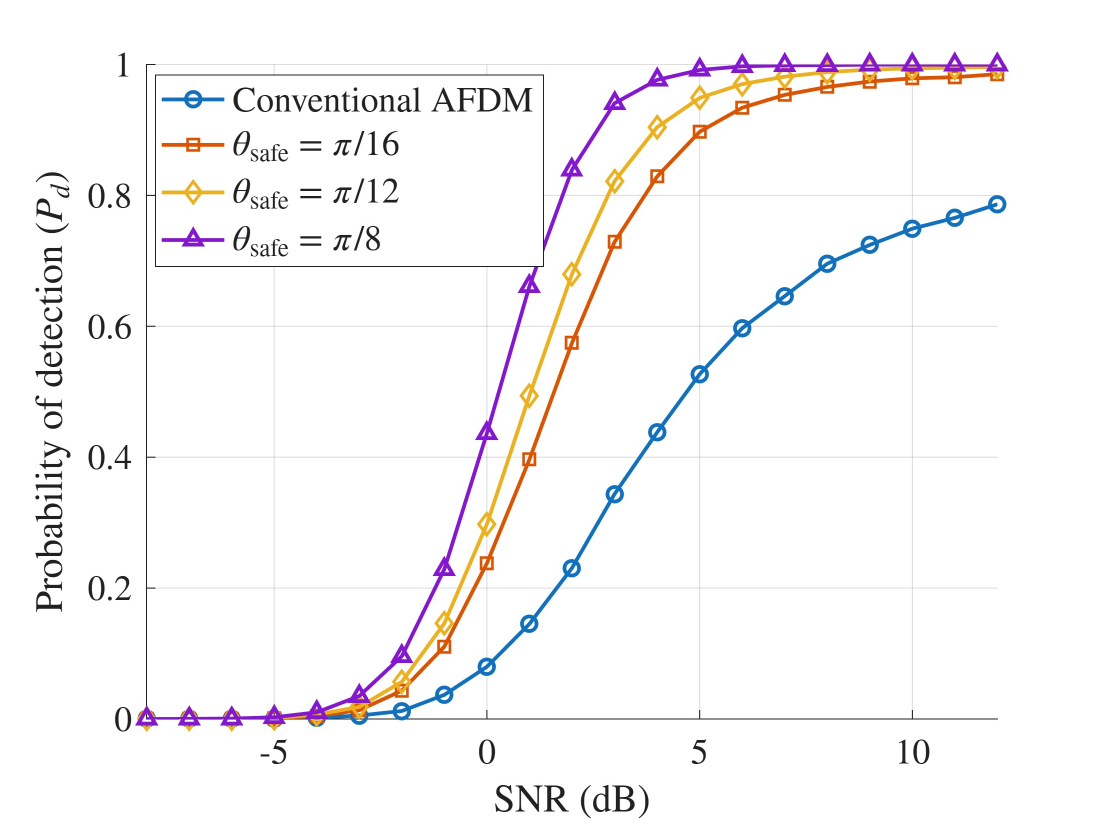}
    \caption{CFAR detection performance.}
    \label{fig:results_cfar}
    \vspace{-6mm}
  \end{figure}


  Fig.~\ref{fig:results_BER} presents the BER performance of $c_2$-perturbed AFDM systems with different values of $\theta_{\rm{safe}}$. The simulations are conducted over a three-path doubly dispersive channel with a carrier frequency of $f_c=4$~GHz and a subcarrier spacing of $\Delta f=15$~kHz. The terminal speed is set as $v=500$~km/h, and we adopt the BPSK constellation. As expected, conventional AFDM achieves better BER performance when $\rm{IBO}=10$~dB. Conversely, when $\rm{IBO}=0$~dB, the proposed method achieves lower BER than conventional AFDM. This is because the nonlinear distortion induced by high PAPR outweighs the phase distortion introduced by the proposed method.

  In Fig.~\ref{fig:results_cfar}, we further characterize the probability of detection $P_d$ versus SNR using a cell-averaging constant false alarm rate detector with the false alarm rate set to $10^{-6}$ for a dual-target sensing scenario, where the delay indices of the two targets are $[10,24]$ and the corresponding normalized radar cross sections are $[1,0.8]$. As shown in Fig.~\ref{fig:results_cfar}, all the proposed $c_2$-perturbation designs outperform conventional AFDM in detection performance, and a higher phase budget leads to a higher $P_d$ given an SNR. In particular, the proposed design with $\theta_{\rm{safe}}=\pi/8$ approaches near-unity detection probability at moderate-to-high SNRs, indicating better sensing robustness under challenging conditions. From Fig.~\ref{fig:results_BER} and Fig.~\ref{fig:results_cfar}, it can be observed that the proposed method enables not only a flexible PAPR vs. WISL trade-off but also a more fundamental communication vs. sensing performance trade-off, with achievable Pareto-efficient operating points.\vspace{-0.1cm}

  \section{Conclusion}
  \label{sec:conclusion}

  This paper presented the design of an ISAC-oriented AFDM waveform by utilizing bounded and subcarrier-wise $c_2$ perturbations at the transmitter. Specifically, we formulated a phase-safety-constrained joint optimization problem for PAPR and WISL, and we derived modulation-aware safe-angle thresholds for rectangular QAM constellations. To address the optimization problem, an NMLS-SPG solver based on analytical gradients is developed. It was demonstrated that our scheme can achieve a controllable trade-off, i.e., the relaxation of the phase budget facilitates PAPR reduction and autocorrelation sidelobe suppression, at the expense of a moderate degradation in BER. Future work will focus on the joint design of the transceiver, including constellation shaping and improved receiver decision algorithms. Furthermore, the integration of deep learning methods will be explored to enhance optimization efficiency and meet the requirements of low-latency communication.

  \bibliographystyle{IEEEtran}
  \bibliography{references}

  \end{document}